\documentclass[10pt, doublecolumn]{IEEEtran}
\usepackage{epsfig,latexsym}
\usepackage{float}
\usepackage{indentfirst}
\usepackage{amsmath}
\usepackage{amssymb}
\usepackage{times}
\usepackage{subfigure}
\usepackage{psfrag}
\usepackage{hyperref}
\usepackage{cite}
\usepackage{lastpage}
\usepackage{fancyhdr}
\usepackage{color}
 \usepackage{amsthm}
\usepackage{bigints}
\newtheorem{Lemma}{Lemma}
\newtheorem{lemma}[Lemma]{$\mathbf{Lemma}$}

\newcounter{problem}
\newcounter{save@equation}
\newcounter{save@problem}
\makeatletter
\newenvironment{problem}
{\setcounter{problem}{\value{save@problem}}%
  \setcounter{save@equation}{\value{equation}}%
  \let\c@equation\c@problem
  \subequations
}
{\endsubequations
  \setcounter{save@problem}{\value{equation}}%
  \setcounter{equation}{\value{save@equation}}%
}
 
\begin{document}%%
\title{ {\Huge Robust Beamforming Design for OTFS-NOMA}}

\author{ Zhiguo Ding, \IEEEmembership{Senior Member, IEEE}  \thanks{ 
  
\vspace{-2em}

 The work of Z. Ding was supported by the UK   Engineering and Physical Sciences Research Council under grant number EP/P009719/2 and by H2020-MSCA-RISE-2015 under grant number 690750.

     Z. Ding
 is    with the School of
Electrical and Electronic Engineering, the University of Manchester, Manchester, UK (email: \href{mailto:zhiguo.ding@manchester.ac.uk}{zhiguo.ding@manchester.ac.uk}.

  }\vspace{-1em}}
 \maketitle
 
\begin{abstract} 
This paper considers the design of beamforming   for orthogonal time frequency space modulation assisted non-orthogonal multiple access (OTFS-NOMA) networks, in which a high-mobility user is sharing the spectrum with multiple low-mobility NOMA users.   In particular, the beamforming design is formulated as an optimization problem whose objective is to  maximize  the low-mobility NOMA users' data rates while guaranteeing that the high-mobility user's targeted data rate can be met.  Both the cases with and without channel state information errors are considered, where low-complexity solutions are developed by applying successive convex approximation  and semidefinite relaxation. Simulation results are also provided to show that the use of the proposed beamforming schemes can yield a significant performance gain over random beamforming. 

%identify   the criteria for the selection of the two proposed SGF mechanisms as well as  appropriate successive interference cancelation decoding orders. 

\end{abstract} \vspace{-1em}
 \section{Introduction}
In conventional non-orthogonal multiple access (NOMA) networks, spectrum sharing among multiple users is encouraged under the condition that users can be distinguished  based on  their channel conditions \cite{NOMAPIMRC,Nomading}. A recent work in \cite{8786203} proposed a new form of NOMA, termed OTFS-NOMA, which applies  orthogonal time frequency space modulation (OTFS) to NOMA and yields an alternative method to  distinguish  users by  their mobility profiles. In particular, OTFS-NOMA is motivated by the fact that conventional OTFS   mainly relies on the delay-Doppler plane, where   a high-mobility user's signals are converted   from the time-frequency plane to the delay-Doppler plane, such that the user can experience time-invariant channel fading    \cite{OTFS2, OTFS, 8424569,8516353,8503182}.  Unlike conventional OTFS, OTFS-NOMA   uses the bandwidth resources available in both the time-frequency plane and the delay-Doppler plane \cite{8786203, otfsnoma1,nomaotfs2}. As a result,  in OTFS-NOMA, high-mobility users can be still served with time-invariant channels in the delay-Doppler plane, and the resources in the time-frequency plane can also been released to those low-mobility users, which improves the overall spectral efficiency.   Such spectrum sharing among the users with different mobility profiles can be particularly important to 5G and beyond communication scenarios, where  some users might be static, e.g., Internet of Things (IoT) sensors, and there might be some users which are moving at very high speeds, e.g., users in a high-speed train. 

In this paper, we focus on  an OTFS-NOMA downlink transmission scenario, in which a high-mobility user  and multiple low-mobility NOMA users share the spectrum and are served simultaneously by a base station. Unlike \cite{8786203},  we assume  that   the base station has multiple antennas and each user has a single antenna. The design of beamforming is considered  in   this paper, where our objective is to maximize  the low-mobility NOMA users' data rates while guaranteeing that  the high-mobility user's targeted data rate can be be met. Because of the users' heterogeneous mobility profiles, we assume that the low-mobility NOMA users' channel state information (CSI) is perfectly known by the base station, but there exist errors for the high-mobility user's CSI. In the presence of these CSI errors, a robust beamforming optimization problem is formulated and solved by applying successive convex approximation (SCA). We note that, in the case with perfect CSI, the formulated robust beamforming design problem can be degraded to a simplified form   which facilitates the development of a more computationally efficient method based on 
semidefinite relaxation (SDR). Computer simulation results are provided to demonstrate that the proposed SCA robust beamforming scheme can efficiently utilize the spatial degrees of freedom and achieve a significant performance gain over the case with a randomly chosen beamformer. In the case with perfect CSI, the developed SDR method can realize a better performance than the SCA-based scheme, where both the proposed schemes outperform random beamforming.

 \section{ System Model } \label{section 2}
 This paper considers an OTFS-NOMA downlink scenario, 
  with  one base station communicating with $(M+1)$   users, denoted by $\text{U}_i$, $0\leq i\leq M$. The base station has $V$ antennas, and each user is equipped with a single antenna.

As in \cite{8786203}, $\text{U}_0$ is assumed to be a high-mobility user and  its $N M$  information bearing  signals, denoted by $x_0[k,l]$, $k\in\{0, \cdots, N-1\}$, $l\in\{0,\cdots, M-1\}$, are placed in the delay-Doppler plane, where $M$ and $N$ denote the OTFS   parameters and define how   the delay-Doppler and time-frequency planes are partitioned, e.g., the time-frequency plane can be  partitioned  by sampling at intervals of $T$ and frequency spacing $\Delta f$.    By using the inverse symplectic finite Fourier transform (ISFFT),    $\text{U}_0$'s signals    are converted from the delay-Doppler plane to     the time-frequency plane   \cite{OTFS}:
\begin{align}
\mathbf{x}_0^{\rm{TF}} =& \mathbf{F}^H_N\otimes \mathbf{F}_M \mathbf{x}_0
,
\end{align}
where $\mathbf{F}_n$ denotes an $n\times n$ discrete Fourier transform
(DFT) matrix and $\mathbf{x}_0$ is an $NM\times 1$ vector collecting all $x_0[k,l]$. Denote $X_0[n,m]$ by the $(nN+m+1)$-th element of $\mathbf{x}_0^{\rm{TF}}$, where $0\leq n \leq N-1$ and $0\leq m\leq M-1$. 
 
 Without using NOMA, i.e., OTFS with orthogonal multiple access (OTFS-OMA), only the high-mobility user, $\text{U}_0$, is served during   $NT$ and $M\Delta f$, whereas other users cannot be admitted to these time slots or frequencies.  The main motivation for using OTFS-NOMA is to create an opportunity for ensuring that  the bandwidth resources, $NT$ and $M\Delta f$, can   be shared between the high-mobility user and additional low-mobility users, by applying the principle of NOMA.  Such spectrum sharing is particularly important   if the high-mobility user has weak channel conditions or   needs to be served with a small data rate only\footnote{We note that even if the high-mobility user's channel conditions are strong and this user wants to be served with a high data rate, channel uncertainties caused by the user's high mobility can still reduce the user's achievable data rate. As a result, it is spectrally inefficient to allow  all the bandwidth resources to be solely occupied by the high-mobility user.}.  Similar to \cite{8786203}, the $M$ NOMA  users, $\text{U}_i$, $1\leq i \leq M$, are assumed to be low-mobility users, and their signals, denoted by $X_i[n,m]$,  are placed directly in the time-frequency plane, and are superimposed with the high-mobility user's   signals, $X_0[n,m]$.    As in \cite{8786203},  we assume that the NOMA users' time-frequency signals are generated as follows:  
\begin{align}\label{indooruser2}
X_i[n,m] =\left\{\begin{array}{ll} x_i(n) & \text{if}\quad m=i-1 \\ 0 &\text{otherwise}\end{array}\right.,
\end{align}
  where $x_i(n)$, for $1\leq i \leq M$ and $0\leq n\leq N-1$, are $\rm{U}_i$'s information bearing signals.  

At its $v$-th antenna, $1\leq v\leq V$, the base station superimposes $\text{U}_0$'s   signals with the NOMA users' signals by using the beamforming coefficient, denoted by $w_v$, as follows~: 
\begin{align}\label{downlink1}
X^v[n,m] =&  w_v\sum^{M}_{i=0} X_i[n,m].
\end{align}
We note that the use of power allocation among the users can further improve the performance of OTFS-NOMA, however, scaling the $(M+1)$ users' signals with different coefficients makes an explicit expression for the  signal-to-interference-plus-noise ratio (SINR) difficult to obtain. Therefore,  the design of beamforming is focused in this paper, where the joint design of beamforming and power allocation is a promising direction for future research but beyond the scope of this paper. 

Following steps similar to those in \cite{8786203,OTFS, OTFS2,8424569},    the received signal  at $\text{U}_i$  in the time-frequency plane  can be modelled  as follows: 
\begin{align}\label{downlink1x}
Y_i[n,m] =& \sum^{V}_{v=1}w_v H^v_i[n,m]X[n,m] +W_i[n,m],
\end{align}
where $W_i(n,m)$ is the white Gaussian noise in  the time-frequency plane, and $H^v_i[n,m]$ denotes the time-frequency channel gain between the $v$-th antenna at the base station and $\text{U}_i$. The users' detection strategies are described in the following subsections. 
\vspace{-1em}

\subsection{  Detecting   the High-Mobility User's Signals} \label{section 4} 
Each user will first try to detect  $\text{U}_0$'s signals in the delay-Doppler plane. In particular, the system model in the delay-Doppler plane at $\text{U}_i$' is given by \cite{8786203}: 
\begin{align}\label{yx model2}
 \mathbf{y}_i = \sum^V_{v=1}w_v\mathbf{H}^v_i \mathbf{x}_0 +\underset{\text{Interference and noise terms}}{\underbrace{ \sum^V_{v=1}w_v\sum^{M}_{q=1} \mathbf{H}_i^v {\mathbf{x}}_q +\mathbf{z}_i}},
\end{align}
where   $\mathbf{H}_i^v$ is an $NM\times NM$ block-circulant matrix,  $\mathbf{y}_i$ denotes $\text{U}_i$'s $NM$ observations in the delay-Doppler plane, $\mathbf{x}_q$ denotes $\text{U}_q$'s   signals in the delay-Doppler plane and $\mathbf{z}_i$ denotes the noise vector. 

 In this paper, the use of a frequency-domain linear equalizer is considered\footnote{ We note that equalization is still carried out in the delay-Doppler plane, not in the time-frequency plane. The terminologies,  frequency-domain equalizers, are used  because the system model in \eqref{yx model2} can be regarded as a model in conventional single carrier cyclic prefix systems, to which various equalization techniques, termed frequency-domain equalizers, have been developed \cite{8786203}.  In addition to the frequency-domain linear equalizer, we note that other types of equalizers can also be used. For example, one can   use a frequency-domain decision feedback equalizer (FD-DFE). As shown in \cite{8786203}, the use of such more advanced equalizers can further improve the performance of OTFS-NOMA, but results in more computational complexity. }. In particular, by applying the detection matrix, $\left( \mathbf{F}_N^H \otimes \mathbf{F}_M\right)\left(\sum^{V}_{v=1}w_v\mathbf{D}_i^v \right)^{-1}\mathbf{F}_N \otimes \mathbf{F}_M^H  $, to  the observation vector $\mathbf{y}_i$,  the received signals  for OTFS-NOMA downlink transmission can be written as follows:  
\begin{align}\label{approach 1 system model 1}
&\left( \mathbf{F}_N^H \otimes \mathbf{F}_M\right) \mathbf{D}_i^{-1}\mathbf{F}_N \otimes \mathbf{F}_M^H   {\mathbf{y}}_i\\\nonumber
 = &  \mathbf{x}_0 +\underset{\text{Interference and noise terms}}{\underbrace{ \ \sum^{M}_{q=1}   {\mathbf{x}}_q  + \left( \mathbf{F}_N^H \otimes \mathbf{F}_M\right) \mathbf{D}_i^{-1}\mathbf{F}_N \otimes \mathbf{F}_M^H{\mathbf{z}}_i}},
\end{align} 
where $\mathbf{D}_i=\sum^{V}_{v=1}w_v\mathbf{D}_i^v $, $\mathbf{D}_i^v $
  is a diagonal matrix whose $(kM+l+1)$-th  main diagonal element is given by  
\begin{align}
D_{i}^{k,l,v}= \sum^{N-1}_{n=0}\sum^{M-1}_{m=0} a_{i,n}^{m,1,v} e^{j2\pi \frac{lm}{M}} e^{-j2\pi \frac{kn}{N}},
\end{align}
for $0\leq k \leq N-1$, $0\leq l \leq M-1$, and 
 $a^{m,1,v}_{i,n}$ is the element located in the $(nM+m+1)$-th row and the first column of $\mathbf{H}_i^v$.  Therefore,  the SINRs for detecting  $x_0[k,l]$,  $0\leq k \leq N-1$ and $0\leq l \leq M-1$, are identical  and   given by 
\begin{align}\label{sinrzf}
\text{SINR}_{i}^{\text{I}} =&  \frac{\rho }{\rho  +\frac{1}{NM}\sum^{N-1}_{\tilde{k}=0}\sum^{M-1}_{\tilde{l}=0} \left|\sum^{V}_{v=1}w_vD_i^{\tilde{k},\tilde{l},v}\right|^{-2}},
\end{align}
where $\rho$ denotes the transmit signal-to-noise ratio (SNR).

\subsection{Detecting   the Low-Mobility NOMA Users' Signals} \label{subsection 2}
Assume that $\text{U}_0$'s   signals can be decoded and  removed  successfully, which means that, in the time-frequency plane, the NOMA users observe the following: 
\begin{align} 
Y_i[n,m]   =&\sum^{V}_{v=1}w_v  H^v_i[n,m] x_{m+1}(n) +W_i[n,m]
.
\end{align}
Therefore,    the signals for $\rm{U}_i$  experience the following SNR: 
\begin{align}\label{snrx}
\text{SNR}_{i}^{\rm{II}} = \rho   \left|\sum^{V}_{v=1}w_v {D}_i^{0,i-1,v}\right|^2,
\end{align}
where ${D}_i^{0,i-1,v}$ is used since it is assumed that the low mobility NOMA users' channels are time invariant.  
 
% \begin{problem}\label{pb:2}
%  \begin{alignat}{2}
%    &  {\mathbf{x}}
%    &\qquad & \sum_{i=1}^nv_ix_i\label{obj:2}\\
%    & \text{subject to}
%    & & \sum_{i=1}^nw_ix_{i}\leqslant W,\label{c:1}\\
%    & & &\!\begin{aligned} (c_1 +c_2 +\dots+c_n&)x_ix_j\leqslant C,\\ & \forall i\in\{1,\ldots,n\}, j\in\{1,\ldots,n\},
%    \end{aligned}\label{c:2}\\
%    & & & x_{ i }\in\{0, 1\},\; \forall i=1\ldots,n.\label{c:3}
%  \end{alignat}
%\end{problem}

Define $\mathbf{h}_{i,l}=\begin{bmatrix}{D}_i^{0,l,1}&\cdots &{D}_i^{0,l,V}\end{bmatrix}^T$, and $\mathbf{g}_{k,l}=\begin{bmatrix}{D}_0^{k,l,1}&\cdots &{D}_0^{k,l,V}\end{bmatrix}^T$. Because of the $\rm{U}_0$'s high mobility, it is assumed that the base station does not have the perfect knowledge of $\rm{U}_0$'s CSI, and this CSI uncertainty is modelled as follows \cite{5164911,6626661}:
\begin{align}
\mathbf{g}_{k,l} = \hat{\mathbf{g}}_{k,l} +\mathbf{e}_{k,l}, 
\end{align}
where $\hat{\mathbf{g}}_{k,l}$ denotes the channel estimates available at the base station and $\mathbf{e}_{k,l}$ denotes the  CSI errors. In particular, it is assumed that the CSI errors are bounded as follows: $\mathbf{e}_{k,l}^H\mathbf{C}_{k,l}\mathbf{e}_{k,l}\leq 1$. As in \cite{5164911,6626661}, we choose $\mathbf{C}_{k,l}=\sigma^{-2}\mathbf{I}_V$, where $\sigma$ denotes the parameter for indicating the accuracy of the channel estimates.  

In this paper, we focus on the robust beamforming design problem which can be formulated as follows:  
 \begin{problem}\label{pb:1}
  \begin{alignat}{2}
\rm{max.} &\qquad {\rm min.} \quad \left\{ \log(1+\text{SINR}_{i}^{\text{II}}), 1\leq i\leq M \right\}\label{obj:1} \\
\rm{s.t.} & \qquad \! \begin{aligned} \underset{\substack{ \mathbf{e}_{ {k}, {l}}^H \mathbf{e}_{ {k}, {l}}\leq \sigma^{2}\\ 0\leq k\leq N-1, 0\leq l \leq M-1
}}{\rm min.}& \log(1+\text{SINR}_{i}^{\text{I}})\geq R_0, \\&\qquad \qquad0\leq i \leq M\end{aligned}\label{st:1}
\\
& \qquad \mathbf{w}^H\mathbf{w}\leq 1,\label{st:2}
  \end{alignat}
\end{problem} 
where  $\mathbf{w}$ denotes  a $V\times1$ beamforming vector and $R_0$ denotes  the high-mobility user's target data rate.  The aim of the optimization problem formulated in \eqref{pb:1}  is explained in the following. The objective function in   \eqref{obj:1} is   the minimum of the $M$ low-mobility NOMA users' data rates,  i.e., the objective of the problem formulated in (P1) is to maximize the low-mobility NOMA users' data rates. \eqref{st:1} is to ensure that the first stage of successive interference cancellation (SIC) is successful at both the low-mobility and high-mobility users, and \eqref{st:2} is to ensure the power normalization. It is important to point out that  the constraints shown in \eqref{st:1} serve two   purposes. One is to ensure that the high-mobility user's  data rate realized by the OTFS-NOMA transmission scheme is larger than this user's target data rate,   $R_0$, i.e., the high-mobility user's target data rate can be guaranteed.    The other is to ensure that   the low-mobility NOMA users can successfully decode the high-mobility user's signals. Without these constraints, it is possible that one low-mobility user fails the first stage of 	SIC, which means that the rate $\log(1+\text{SINR}_{i}^{\text{II}})$ might not be achievable.

 \section{Low-Complexity Beamforming Design}
 In this section, we first focus on  the case $\sigma\neq 0$, where we will convert the  objective function into a convex form, then recast the robust beamforming design problem to an equivalent  form without CSI errors $\mathbf{e}_{k,l}$, and finally apply SCA to obtain a low-complexity solution. In addition, the case with perfect CSI will also be studied, where an optimal solution based on SDR can be obtained for a special case.  
 
\subsection{Robust Beamforming Design}
By using  the expressions of the SINRs in \eqref{sinrzf}  and the SNRs in   \eqref{snrx} and with some algebraic manipulations,    the robust beamforming design     problem can be recasted   as follows:
 \begin{problem}\label{pb:3}
  \begin{alignat}{2}
\rm{max.} &\qquad {\rm min.} \quad \left\{ |\mathbf{w}^H\mathbf{h}_{i,i-1}|^2, 1\leq i\leq M \right\} \\
\rm{s.t.} &   \qquad  
\!\begin{aligned}  \sum^{N-1}_{ {k}=0}\sum^{M-1}_{ {l}=0}  \underset{  \mathbf{e}_{ {k}, {l}}^H \mathbf{e}_{ {k}, {l}}\leq \sigma^{2}}{{\rm max.}} (\mathbf{w}^H(\hat{\mathbf{g}}_{ {k}, {l}}+\mathbf{e}_{ {k}, {l}})\\\times(\hat{\mathbf{g}}_{ {k}, {l}}+\mathbf{e}_{ {k}, {l}})^H\mathbf{w})^{-1}  \leq \epsilon 
    \end{aligned} \label{st 30}
\\
& \qquad   \sum^{M-1}_{ {l}=0} \frac{1}{\mathbf{w}^H\mathbf{h}_{   {i,l}}\mathbf{h}_{  {i,l}}^H\mathbf{w}  }\leq \epsilon_1 , 1\leq i \leq M\label{st 31}
\\
& \qquad \mathbf{w}^H\mathbf{w}\leq 1,\label{st 32}
  \end{alignat}
\end{problem} 
$\eta=2^{R_0}-1$, $\epsilon=\rho NM(\eta^{-1}-1)$ and $\epsilon_1=\rho M(\eta^{-1}-1)$.

In order to remove the minimization operator in the objective function, we recast \eqref{pb:3} as follows:
 \begin{problem}\label{pb:41}
  \begin{alignat}{2}
\rm{max.} &\qquad    z \\
\rm{s.t.} & \qquad  \label{410} - |\mathbf{w}^H\mathbf{h}_{i,i-1}|^2\leq -z, 1\leq i\leq M  \\ &\qquad   \eqref{st 30}, \eqref{st 31}, \eqref{st 32}.\nonumber
  \end{alignat}
\end{problem}

Although the objective function in \eqref{pb:41} becomes an affine function, the constraint  in \eqref{410} is not convex, since the left-hand side of the inequality is a concave function. Therefore, we further recast  \eqref{pb:41} as follows:
 \begin{problem}\label{pb:42}
  \begin{alignat}{2}
\rm{min.} &\qquad  t   \\
\rm{s.t.} & \qquad  \label{411}  \frac{1}{|\mathbf{w}^H\mathbf{h}_{i,i-1}|^2}\leq  t, 1\leq i\leq M  \\ &\qquad   \eqref{st 30}, \eqref{st 31}, \eqref{st 32}.\nonumber
  \end{alignat}
\end{problem}   
The left-hand side of the constraint in \eqref{411} is a convex function, as can be shown  in the following. Define $\tilde{f}_{i,l}(\mathbf{w})=   \frac{1}{ \mathbf{w}^H  {\mathbf{h}}_{ {i}, {l}} {\mathbf{h}}_{ {i}, {l}}^H\mathbf{w}}   $, where its first order derivative is given by 
 \begin{align}\label{de1}
&\bigtriangledown \tilde{f}_{i,l}(\mathbf{w})= \frac{-2  {\mathbf{h}}_{ {i}, {l}} {\mathbf{h}}_{ {i}, {l}}^H\mathbf{w}  }{     \left(\mathbf{w}^H  {\mathbf{h}}_{ {i}, {l}} {\mathbf{h}}_{ {i}, {l}}^H\mathbf{w}\right)^{2}   } ,
 \end{align} 
 and its second order derivative is given by
  \begin{align}\nonumber
&\bigtriangledown^2 \tilde{f}_{i,l}(\mathbf{w})=     \frac{6   {\mathbf{h}}_{ {i}, {l}}        {\mathbf{h}}_{ {i}, {l}}  ^H   }
{\left(\mathbf{w}^H  {\mathbf{h}}_{ {i}, {l}} {\mathbf{h}}_{ {i}, {l}}^H\mathbf{w}\right)^{2} }
,
 \end{align}   
 which is positive semidefinite. Therefore,  the constraint in \eqref{411} is in a convex form. 

In oder to remove the CSI errors, $\mathbf{e}_{k,l}$, from the optimization problem, we first   define the following feasibility problem:
 \begin{problem}\label{pb:x}
  \begin{alignat}{2}
\rm{Find} &\quad \mathbf{e}_{k,l} \\
\rm{s.t.} & \quad     \mathbf{w}\mathbf{w}^H \mathbf{e}_{ {k}, {l}} =- \mathbf{w}\mathbf{w}^H\hat{\mathbf{g}}_{ {k}, {l}}
\\ &\quad 
\mathbf{e}_{ {k}, {l}}^H \mathbf{e}_{ {k}, {l}}\leq \sigma^{2}     .\end{alignat}
\end{problem} 
The following lemma is provided to simplify  the optimization problem shown in \eqref{pb:42} and facilitate the application of SCA and SDR.  

\begin{lemma}
The robust beamforming  optimization problem in \eqref{pb:42} can be equivalently recast as follows: 
 \begin{problem}\label{pb:4}
  \begin{alignat}{2}
\rm{min.} &\quad t \\
\rm{s.t.} & \quad   \frac{1}{ \mathbf{w}^H\mathbf{h}_{i,i-1}\mathbf{h}_{i,i-1}^H\mathbf{w}} \leq t, 1\leq i\leq M  \label{stx400}
\\ &\quad   \sum^{N-1}_{ {k}=0}\sum^{M-1}_{ {l}=0} \frac{1}{  \left(   \left(\mathbf{w}^H \hat{\mathbf{g}}_{ {k}, {l}}\hat{\mathbf{g}}_{ {k}, {l}}^H\mathbf{w}\right)^{\frac{1}{2}}-\left(\sigma^{2}  \mathbf{w}^H\mathbf{w}\right)^{\frac{1}{2}} \right)^2 }\leq \epsilon  \label{stx40}
\\
& \quad   \sum^{M-1}_{ {l}=0} \frac{1}{\mathbf{w}^H\mathbf{h}_{   {i,l}}\mathbf{h}_{  {i,l}}^H\mathbf{w}  }\leq \epsilon_1 , 1\leq i \leq M \label{stx41}
\\
& \quad \mathbf{w}^H\mathbf{w}\leq 1, \label{stx42}
  \end{alignat}
\end{problem} 
if an optimal  solution of  \eqref{pb:4} can ensure that  the   problem shown in \eqref{pb:x} is infeasible for any $\hat{\mathbf{g}}_{k,l}$, $k\in \{0, N-1\}$ and $l\in \{0, \cdots, M-1\}$. Otherwise,  the robust beamforming  optimization problem in \eqref{pb:42} is  infeasible. 

\end{lemma}
\begin{proof}
Please refer to the appendix. 
\end{proof}
{\it Remark 1:} The fact that strong CSI errors results in the infeasibility situation can be explained in the following. With strong CSI errors, it is very likely to have 
\begin{align}
\underset{  \mathbf{e}_{ {k}, {l}}^H \mathbf{e}_{ {k}, {l}}\leq \sigma^{2}}{{\rm min.}} \mathbf{w}^H(\hat{\mathbf{g}}_{ {k}, {l}}+\mathbf{e}_{ {k}, {l}}) (\hat{\mathbf{g}}_{ {k}, {l}}+\mathbf{e}_{ {k}, {l}})^H\mathbf{w} =0,
\end{align}
which leads to the situation that the constraint in \eqref{st 30} can never be satisfied, e.g., the problem is infeasible. We also note that an optimal solution obtained by solving   \eqref{pb:4} is not necessarily an optimal solution for \eqref{pb:42}. Or in other words, only if an optimal solution of \eqref{pb:4} fails the feasibility check for \eqref{pb:x}, one can claim that this optimal solution of \eqref{pb:4} is  also optimal  for \eqref{pb:42}.

{\it Remark 2:} The formulation in \eqref{pb:4} is general and can also be applicable to the case without CSI errors. In particular, by setting   $\sigma=0$, one can easily verify that the formulation in \eqref{pb:4} is indeed applicable  to the one without CSI errors.  

We note that problem \eqref{pb:4}  is still not convex, mainly due to the non-convex constraint in \eqref{stx40}. In the following, SCA is applied to obtained a low-complexity suboptimal solution.  

  Define $f_{k,l}(\mathbf{w})= \left(   \left(\mathbf{w}^H \hat{\mathbf{g}}_{ {k}, {l}}\hat{\mathbf{g}}_{ {k}, {l}}^H\mathbf{w}\right)^{\frac{1}{2}}-\left(\sigma^{2} \mathbf{w}^H\mathbf{w}\right)^{\frac{1}{2}} \right)^{-2}$, where its first-order derivative is given by
 \begin{align}\label{de2}
&\bigtriangledown f_{k,l}(\mathbf{w})=\\\nonumber &\frac{-2\left[ \left(\mathbf{w}^H \hat{\mathbf{g}}_{ {k}, {l}}\hat{\mathbf{g}}_{ {k}, {l}}^H\mathbf{w}\right)^{-\frac{1}{2}} \hat{\mathbf{g}}_{ {k}, {l}}\hat{\mathbf{g}}_{ {k}, {l}}^H\mathbf{w}-\sigma \left( \mathbf{w}^H\mathbf{w}\right)^{-\frac{1}{2}} \mathbf{w}\right]}{  \left(   \left(\mathbf{w}^H \hat{\mathbf{g}}_{ {k}, {l}}\hat{\mathbf{g}}_{ {k}, {l}}^H\mathbf{w}\right)^{\frac{1}{2}}-\sigma\left(  \mathbf{w}^H\mathbf{w}\right)^{\frac{1}{2}} \right)^3 } .
 \end{align}  
 
 By  applying the Taylor expansion to the constraints in \eqref{pb:4} at a feasible point $\mathbf{w}_0$ and also using  $\bigtriangledown \tilde{f}_{i,l}(\mathbf{w})$ in \eqref{de1} and $\bigtriangledown f_{k,l}(\mathbf{w})$ in \eqref{de2}, the optimization problem in \eqref{pb:4} can be approximated as follows:
 \begin{problem}\label{pb:5x}
  \begin{alignat}{2}
\rm{min.} &\quad t \\
\rm{s.t.} & \quad  \!\begin{aligned} \tilde{f}_{i,i-1}(\mathbf{w_0}) +\left[\bigtriangledown \tilde{f}_{i,i-1}(\mathbf{w}_0)\right]^H&\left(\mathbf{w}-\mathbf{w}_0\right)    \leq t,\\& 1\leq i\leq M\end{aligned}  
\\   &\quad   \sum^{N-1}_{ {k}=0}\sum^{M-1}_{ {l}=0}  \left( f_{k,l}(\mathbf{w_0}) +\left[\bigtriangledown f_{k,l}(\mathbf{w}_0)\right]^H\left(\mathbf{w}-\mathbf{w}_0\right)  \right)\leq \epsilon 
\\ 
& \quad \!\begin{aligned} & \sum^{M-1}_{ {l}=0} \left( \tilde{f}_{i,l}(\mathbf{w_0}) +\left[\bigtriangledown \tilde{f}_{i,l}(\mathbf{w}_0)\right]^H\left(\mathbf{w}-\mathbf{w}_0\right)  \right)\leq \epsilon_1 ,\\ &\qquad\qquad\qquad\qquad\qquad\qquad\qquad  1\leq i \leq M\end{aligned}
\\
& \quad \mathbf{w}^H\mathbf{w}\leq 1,
  \end{alignat}
\end{problem} 
to which a straightforward  iterative SCA algorithm facilitated by convex optimization solvers, such as CVX, can be applied  to find a solution \cite{Boyd}. 

{\it Remark 3:}  We note that the problem in \eqref{pb:5x} is only an approximated form of the original problem in \eqref{pb:4}, which means that the solution obtained by SCA is only a suboptimal solution of \eqref{pb:4}. 

{\it Remark 4:} One challenge for the implementation of SCA is to find a feasible   $\mathbf{w}_0$ which is required  for the SCA initialization.  However, given the large number of constraints in  \eqref{pb:4}, it is difficult to find a feasible $\mathbf{w}_0$. For the simulations carried out for the paper,   we simply use a randomly generated  $\mathbf{w}_0$. With this randomly generated $\mathbf{w}_0$, the use of SCA can still yield a solution which outperforms the benchmarking scheme, as shown in the next section. 

\vspace{-1em}
\subsection{Beamforming Design Without CSI Errors}
When $\sigma=0$, the problem in \eqref{pb:4} is degraded to a simplified form without CSI errors, to which SCA can be still applicable. We note that, because of the simplified form, SDR also becomes applicable \cite{5447068}, where   the advantage of using SDR is to avoid the iterations for updating $\mathbf{w}_0$ as in SCA. Therefore,  the complexity of SDR can be much smaller than that of SCA. In addition, the use of SDR can also avoid the challenging issue about generating $\mathbf{w}_0$, as discussed in Remark 4.  

By applying the principle of SDR, the problem in \eqref{pb:4} can be recasted to the following equivalent form:
 \begin{problem}\label{pb:6x}
  \begin{alignat}{2}
\rm{min.} &\quad t \\
\rm{s.t.} & \quad   \frac{1}{ \rm{tr} \{\mathbf{W}\mathbf{H}_{i,i-1} \}} \leq t, 1\leq i\leq M  
\\ &\quad   \sum^{N-1}_{ {k}=0}\sum^{M-1}_{ {l}=0} \frac{1}{  \rm{tr}\{  \mathbf{W}  {\mathbf{G}}_{ {k}, {l}}  \} }\leq \epsilon 
\\
& \quad   \sum^{M-1}_{ {l}=0} \frac{1}{\rm{tr}\{\mathbf{W}\mathbf{H}_{   {i,l}}\}  }\leq \epsilon_1 , 1\leq i \leq M
\\
& \quad \rm{tr}\{\mathbf{W}\}\leq 1
\\
& \quad  \mathbf{W}\succeq 0
\\
& \quad  \rm{rank}\{\mathbf{W}\} = 1,
  \end{alignat}
\end{problem} 
where $\mathbf{H}_{i,l}=\mathbf{h}_{i,l}\mathbf{h}_{i,l}^H$, ${\mathbf{G}}_{ {k}, {l}}={\mathbf{g}}_{ {k}, {l}} {\mathbf{g}}_{ {k}, {l}}^H$ and $\rm{tr}\{\cdot\}$ denotes the trace operator. 

The optimization problem in \eqref{pb:6x} is still not convex, mainly due to those constraints containing $ \frac{1}{ \rm{tr} \{\mathbf{W}\mathbf{H}_{k,l} \}}$ and $ \frac{1}{ \rm{tr} \{\mathbf{W}\mathbf{G}_{k,l} \}}$. By introducing auxiliary variables, the problem in \eqref{pb:6x} can be converted into the   equivalent form shown in the following: 
  \begin{problem}\label{pb:7x}
  \begin{alignat}{2}
\rm{min.} &\quad t \\
\rm{s.t.} & \quad   \frac{1}{ x_{i,i-1}} \leq t, 1\leq i\leq M  \label{st 70}
\\ &\quad   \sum^{N-1}_{ {k}=0}\sum^{M-1}_{ {l}=0} \frac{1}{  y_{k,l} }\leq \epsilon \label{st 71}
\\
& \quad   \sum^{M-1}_{ {l}=0} \frac{1}{x_{i,l}  }\leq \epsilon_1 , 1\leq i \leq M
\\ &\quad \rm{tr}\{  \mathbf{W}  {\mathbf{G}}_{ {k}, {l}}  \}=y_{k,l}, 0\leq k \leq N-1, 0\leq l\leq M-1,
\\ &\quad \rm{tr}\{\mathbf{W}\mathbf{H}_{   {i,l}}\}=x_{i,l}, 1\leq i \leq M, 0\leq l\leq M-1,
\\ &\quad  y_{k,l}\geq0, 0\leq k \leq N-1, 0\leq l\leq M-1,
\\ &\quad  x_{i,l}\geq 0, 1\leq i \leq M, 0\leq l\leq M-1,
\\
& \quad \rm{tr}\{\mathbf{W}\}\leq 1
\\
& \quad  \mathbf{W}\succeq 0
\\
& \quad  \rm{rank}\{\mathbf{W}\} = 1.
  \end{alignat}
\end{problem} 
 
We note that all the constraints are either affine or  convex, except the rank-one constraint. Take the constraint \eqref{st 71} as an example. Provided $y_{k,l}\geq0$, $\frac{1}{  y_{k,l} }$ is convex, and hence the left-hand side of the inequality, \eqref{st 71}, is also convex since it   is the sum of those convex functions, $\frac{1}{  y_{k,l} }$. By removing the rank-one constraint, the problem in \eqref{pb:7x} can be efficiently solved by using CVX.

{\it Remark 5:} Simulation results indicate that the rank of the SDR solution obtained for the problem shown in \eqref{pb:7x} is larger than one, and  therefore,  the Gaussian randomization method is needed, which means that  the obtained SDR solution is just a suboptimal solution \cite{5447068}.  

{\it Remark 6:} For a special case which is to maximize a single NOMA user's data rate,   the problem in \eqref{pb:7x} can be simplified as follows:
 \begin{problem}\label{pb:8x}
  \begin{alignat}{2}
\rm{min.} &\quad x_{1,0} \\
\rm{s.t.} & \quad     \rm{P9c-P9j}.\nonumber
  \end{alignat}
\end{problem} 
Simulation results show that the rank of the solution for the problem in \eqref{pb:8x} is always one, which means that the SDR solution is optimal for this single-user special case.

\begin{figure}[t] \vspace{-2em}
\begin{center}\subfigure[$V=4$]{\label{fig1a}\includegraphics[width=0.34\textwidth]{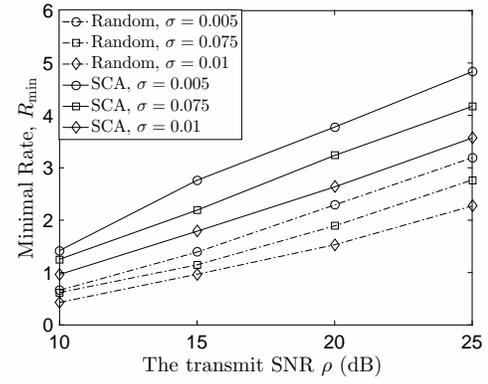}}
\subfigure[ $V=8$]{\label{fig1b}\includegraphics[width=0.34\textwidth]{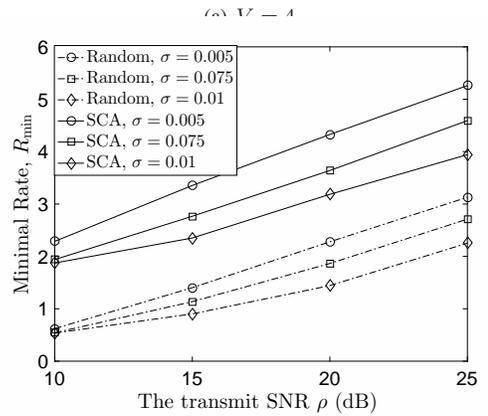}} \vspace{-0.5em}
\end{center}\vspace{-1em}
 \caption{The performance of MISO-OTFS-NOMA with CSI errors. $M=N=8$ and $R_0=0.5$ BPCU.  }\label{fig 1}\vspace{-2em}
\end{figure}

\section{Simulation Results}
In this section, the performance of the developed beamforming schemes  is studied by using computer simulations.   A simple two-path channel model is used. In particular, for $\rm{U}_0$, the  delay-Doppler indices for the two channel taps are $(0,0)$ and $(1,1)$, respectively.   With the subchannel spacing   $\Delta f = 2$ kHz,  the  maximal speed corresponding to the largest Doppler shift $ 250 $ Hz is $67.5$ km/h if the carrier frequency is $f_c=4$ GHz. There is no Doppler shift for the NOMA users. In Fig. \ref{fig 1}, the performance of the robust beamforming design is studied, where $R_{\min}={\rm min.}   \left\{ \log(1+\text{SINR}_{i}^{\text{II}}), 1\leq i\leq M \right\}$ and the random beamforming scheme is used as a benchmarking scheme. As can be observed from the figure,  the proposed low-complexity beamforming design can provide a significant performance gain over the random beamforming scheme. Furthermore, the performance of the proposed SCA  scheme is improved by reducing $\sigma$, i.e., the channel estimates become more accurate.  We note that the performance of the random beamforming scheme is also affected by the choices of $\sigma$, which can be explained in the following. For a randomly chosen $\mathbf{w}$, it is still possible for this $\mathbf{w}$ to make the problem shown in \eqref{pb:x} feasible, if $\sigma$ is large. If this event occurs, we set $R_{\min}=0$, since this event  means that the high-mobility user's targeted data rate cannot be met, as discussed in Remark~1. 

Comparing Fig. \ref{fig1a} to Fig. \ref{fig1b}, one can also observe that the performance of the proposed beamforming design is improved by increasing $V$, whereas the performance of random beamforming stays the same. This is due to the fact that the proposed beamforming scheme can more effectively use the spatial degrees of freedom than random beamforming.  Fig. \ref{fig 2} shows the performance of the proposed beamforming design for the case without CSI errors. Comparing Fig. \ref{fig 2} to Fig. \ref{fig 1}, one can observe that removing the effects of CSI errors improves the performance of the proposed scheme, which is expected. Recall that in the case without CSI errors, the SDR method is also applicable. As can be observed from Fig. \ref{fig2a}, the SDR method outperforms the SCA method, even though the rank of the SDR solution is not one and the Gaussian randomization method has to be applied.  This performance loss of SCA might be due to the used Taylor expansion, since the problem shown in \eqref{pb:5x} is just an approximation to the original problem considered in \eqref{pb:4}.    Fig. \ref{fig2b} shows that, for the single-user case, the performance of the proposed beamforming schemes achieve the same performance. It is worth pointing out that the  rank of the SDR solution for this case is always one, which means that SDR solution is optimal. Therefore, Fig. \ref{fig2b} indicates that the SCA solution is also optimal for the single-user   case.   Furthermore, we note that both the SDR and SCA based schemes outperform  the random beamforming scheme, as can be observed from the two subfigures in Fig. \ref{fig 2}.

\begin{figure}[t] \vspace{-2em}
\begin{center}\subfigure[Multi-User Case]{\label{fig2a}\includegraphics[width=0.34\textwidth]{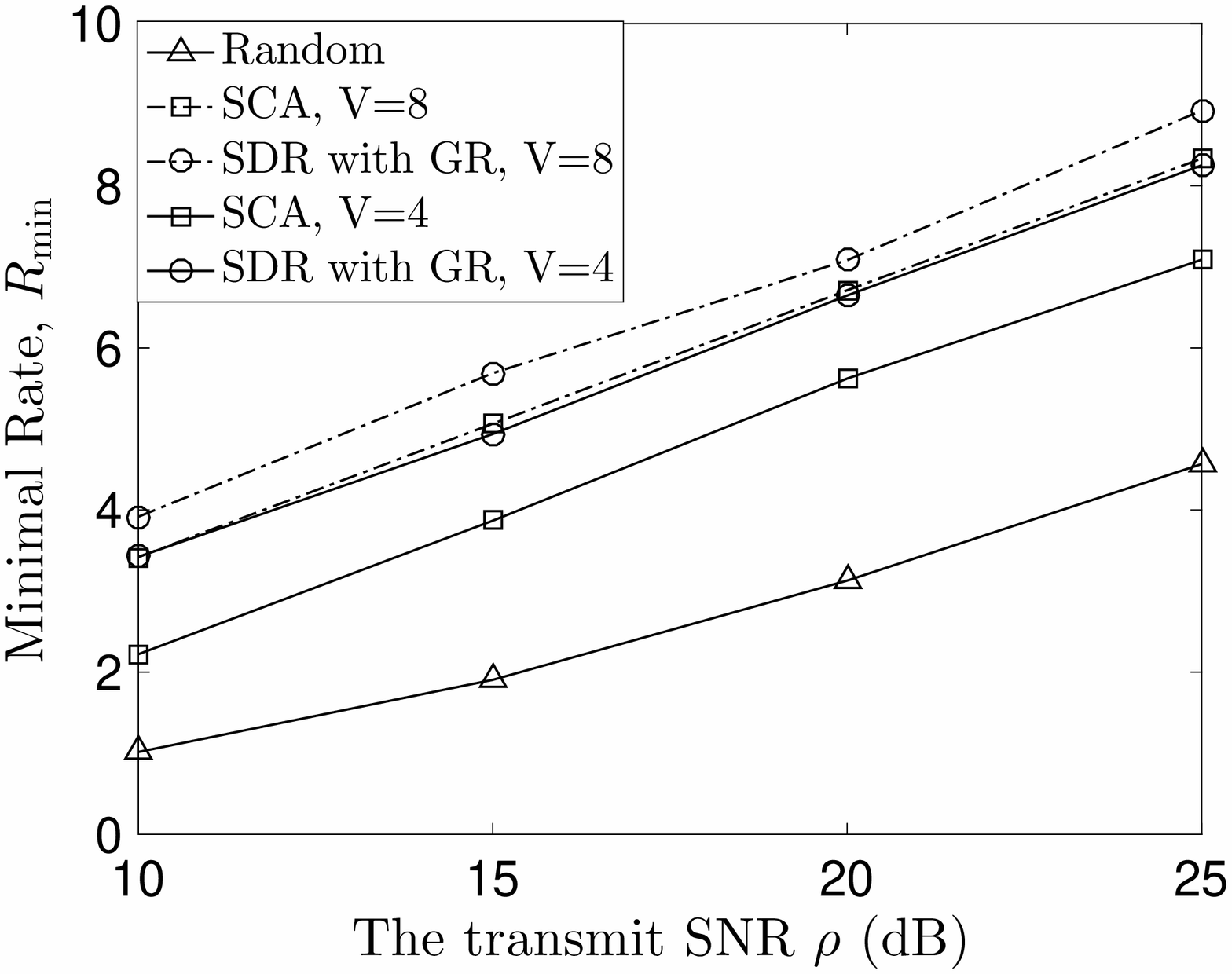}}
\subfigure[Single-User Case]{\label{fig2b}\includegraphics[width=0.34\textwidth]{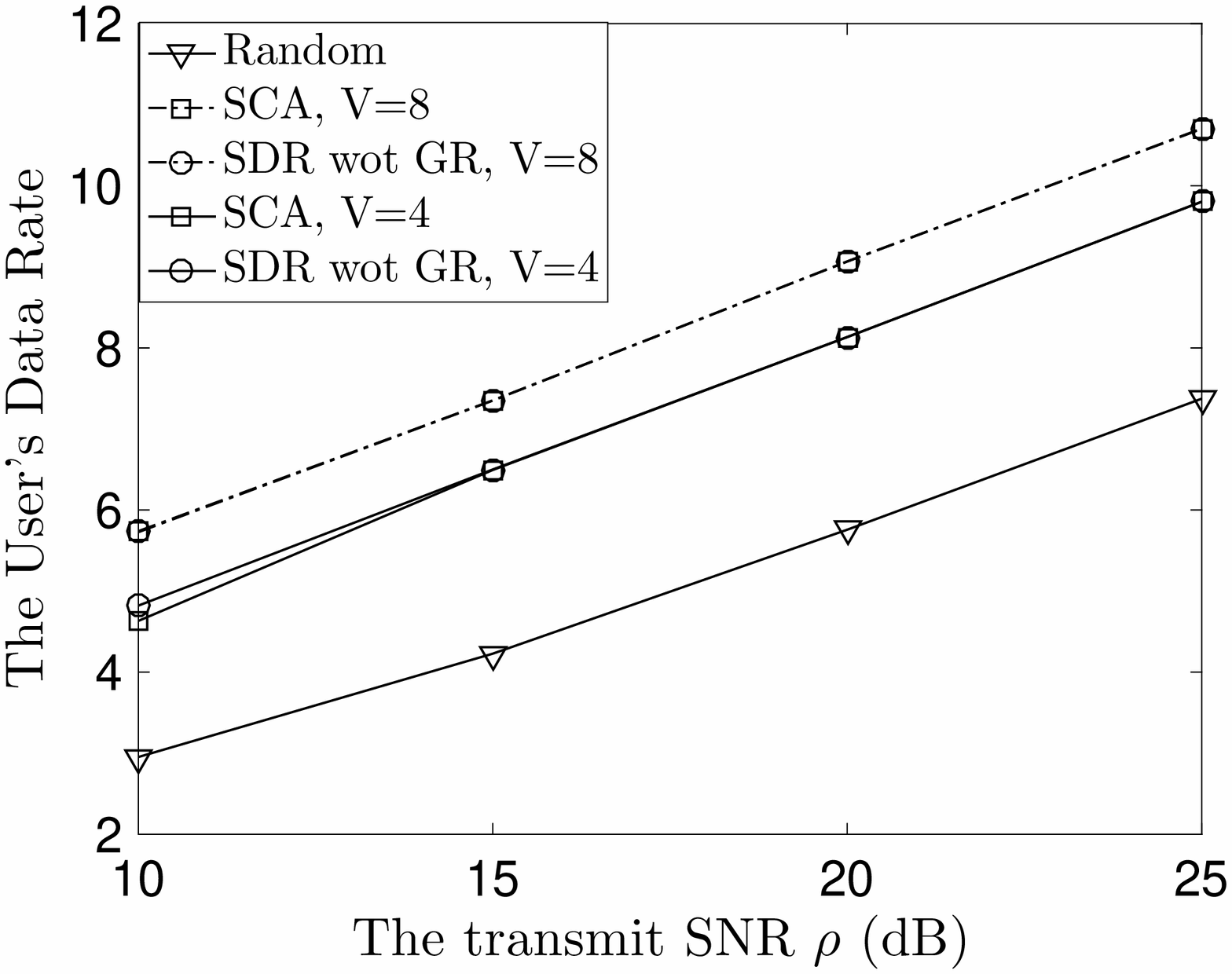}} \vspace{-0.5em}
\end{center}\vspace{-1em}
 \caption{The performance of MISO-OTFS-NOMA without channel estimation errors. $M=N=8$ and $R_0=0.5$ BPCU.  }\label{fig 2}\vspace{-2em}
\end{figure}

\vspace{-0.5em}
\section{Conclusions}
In this paper, we have considered the design of   beamforming   for OTFS-NOMA downlink transmission.  The objective of the considered  beamforming design is to    maximize the  low-mobility NOMA users' data rates while guaranteeing that the high-mobility user's targeted data rate can  be met.  Both the cases with and without CSI errors have been considered, where low-complexity solutions were obtained by applying SCA and SDR. Simulation results have also been provided to show that the use of the proposed beamforming schemes can yield a significant performance gain over random beamforming.  In this paper, we focused on  the downlink transmission scenario, in which more than one high-mobility user can   be accommodated. In particular, we can use the   resource blocks in the delay-Doppler plane to serve multiple high-mobility users. The multiple high-mobility users' signals cause interference to each other in the time-frequency plane, but the use of equalizers   can  ensure that there is no interference in the delay-Doppler plane. As a result, the   optimization problem formulated in this paper can be straightforwardly extended to the case with multiple high-mobility users, by adding more constraints about the high-mobility users' target data rates.  However,  in the uplink scenario with multiple high-mobility users, one high-mobility user's signals can cause interference to the other users' signals in both the time-frequency and delay-Doppler planes. The design for robust OTFS-NOMA transmission for such a challenging uplink scenario is an important topic for future research.
\vspace{-0.5em}
%%%%%%%%%%%%%%%%%%%%%%%%%%%%%%%%%%%%%%%%%%%%%%%%%%%%%%%%%%%%%%%%%%%%%%
\appendix 
 
The robust beamforming optimization problem in \eqref{pb:3} can be first rewritten as follows: 
 \begin{problem}\label{pb:6}
  \begin{alignat}{2}
\rm{max.} &\quad t \\
\rm{s.t.} & \quad    \mathbf{w}^H\mathbf{h}_{i,i-1}\mathbf{h}_{i,i-1}^H\mathbf{w} \geq t, 1\leq i\leq M  
\\ &\quad \label{st 61}\!\begin{aligned}  \sum^{N-1}_{ {k}=0}\sum^{M-1}_{ {l}=0}   (\underset{  \mathbf{e}_{ {k}, {l}}^H \mathbf{e}_{ {k}, {l}}\leq \sigma^{2}}{{\rm min.}} \mathbf{w}^H(\hat{\mathbf{g}}_{ {k}, {l}}+\mathbf{e}_{ {k}, {l}})  \\ \left.(\hat{\mathbf{g}}_{ {k}, {l}}+\mathbf{e}_{ {k}, {l}})^H\mathbf{w}\right)^{-1}  \leq \epsilon \end{aligned}
\\\nonumber 
& \quad  \eqref{st 31},  \eqref{st 32}.
  \end{alignat}
\end{problem} 

The optimization problem contained in  the   constraint   \eqref{st 61} can be rewritten as follows:   \begin{align}\nonumber
 {\rm min.}&\qquad  \mathbf{e}_{ {k}, {l}}^H\mathbf{w}\mathbf{w}^H\mathbf{e}_{ {k}, {l}}+2 \mathcal{R}\left\{ \hat{\mathbf{g}}_{ {k}, {l}}^H\mathbf{w}\mathbf{w}^H\mathbf{e}_{ {k}, {l}}\right\}+ \hat{\mathbf{g}}_{ {k}, {l}}^H\mathbf{w}\mathbf{w}^H \hat{\mathbf{g}}_{ {k}, {l}}\\  {\rm s.t.} &\qquad \mathbf{e}_{ {k}, {l}}^H \mathbf{e}_{ {k}, {l}}\leq \sigma^{2}   , \label{eq48}
\end{align}  
which  is convex as it is a quadratically constrained quadratic program (QCQP). Because   the rank of $\mathbf{w}\mathbf{w}^H$  is one, the considered QCQP can be solved with a closed-form solution as shown in the following. 

  By applying the Karush-Kuhn-Tucker  (KKT) conditions, the optimal solution in \eqref{eq48} can be found by solving the following equations \cite{Boyd}: 
  \begin{eqnarray}\label{kkt}
\left\{\begin{array}{ll}\left(\mathbf{w}\mathbf{w}^H+\lambda \mathbf{I}_V \right)\mathbf{e}_{ {k}, {l}} + \mathbf{w}\mathbf{w}^H\hat{\mathbf{g}}_{ {k}, {l}}  =0\\
\lambda(\mathbf{e}_{ {k}, {l}}^H \mathbf{e}_{ {k}, {l}}- \sigma^{2} )=0\\
\mathbf{e}_{ {k}, {l}}^H \mathbf{e}_{ {k}, {l}}\leq \sigma^{2}  
\end{array}
\right.,
\end{eqnarray}  
where $\lambda$ is the Lagrange multiplier. Depending on the choice of $\lambda$, the optimal solution can be obtained differently as shown in the following. 
\vspace{-1em}
\subsection{For the case   $\lambda=0$} We note that, if $\lambda=0$, the KKT conditions can be simplified  as follows:
  \begin{eqnarray}\label{kktx}
\left\{\begin{array}{ll} \mathbf{w}\mathbf{w}^H \mathbf{e}_{ {k}, {l}} =- \mathbf{w}\mathbf{w}^H\hat{\mathbf{g}}_{ {k}, {l}}  \\ 
\mathbf{e}_{ {k}, {l}}^H \mathbf{e}_{ {k}, {l}}\leq \sigma^{2} 
\end{array}
\right..
\end{eqnarray}

Suppose that there is a solution to   \eqref{kktx}, i.e., the optimization problem in \eqref{pb:x} is feasible, which means that the objective in \eqref{eq48} becomes zero as shown in the following: 
 \begin{align}
  \mathbf{e}_{ {k}, {l}}^H\left(\mathbf{w}\mathbf{w}^H\hat{\mathbf{g}}_{ {k}, {l}}+ \mathbf{w}\mathbf{w}^H\mathbf{e}_{ {k}, {l}}\right)+\hat{\mathbf{g}}_{ {k}, {l}}^H\left(\mathbf{w}\mathbf{w}^H \hat{\mathbf{g}}_{ {k}, {l}}\right.\\\nonumber\left.+ \mathbf{w}\mathbf{w}^H\mathbf{e}_{ {k}, {l}}\right)
 =0.
\end{align}  
As a result, the left-hand side of the inequality   in \eqref{st 61} becomes infinite, which means that the constraint, \eqref{st 61}, can never be satisfied. Or equivalently, the original beamforming optimization  problem in \eqref{pb:42} is infeasible. 
\vspace{-1em}

\subsection{For the case $\lambda \neq 0$}  The KKT conditions shown in \eqref{kkt} can be simplified as follows:
  \begin{eqnarray}\label{kktxxx}
\left\{\begin{array}{ll}\left(\mathbf{w}\mathbf{w}^H+\lambda  \mathbf{I}_V\right)\mathbf{e}_{ {k}, {l}} =- \mathbf{w}\mathbf{w}^H\hat{\mathbf{g}}_{ {k}, {l}}  \\
 \mathbf{e}_{ {k}, {l}}^H \mathbf{e}_{ {k}, {l}}- \sigma^{2} =0 
\end{array}
\right..
\end{eqnarray}  
The fact that the rank of $\mathbf{w}\mathbf{w}^H$ is one can be used to find a closed-form solution for the KKT conditions, as shown in the following.
To solve $\left(\mathbf{w}\mathbf{w}^H+\lambda  \mathbf{I}_V\right)\mathbf{e}_{ {k}, {l}} =- \mathbf{w}\mathbf{w}^H\hat{\mathbf{g}}_{ {k}, {l}} $, we decompose   $(\mathbf{w}\mathbf{w}^H+\lambda  \mathbf{I}_V)$ as follows: 
\begin{align}\label{decom}
\mathbf{w}\mathbf{w}^H+\lambda  \mathbf{I}_V = \mathbf{U}\boldsymbol \Lambda \mathbf{U}^H,
\end{align}
where 
$\mathbf{U}=\begin{bmatrix}\frac{\mathbf{w}}{|\mathbf{w}^H\mathbf{w}|}&\mathbf{w}_{\perp}^{1}&\cdots &\mathbf{w}_{\perp}^{V-1}\end{bmatrix}$, $\mathbf{w}_{\perp}^{v}$ is a normalized vector orthogonal to $\mathbf{w}$, and $\boldsymbol \Lambda =\text{diag}\{\mathbf{w}^H\mathbf{w}+\lambda,\lambda, \cdots, \lambda\}$.

By applying this decomposition, the first condition  in \eqref{kktxxx} can be rewritten as follows:
\begin{align}
\mathbf{U}\boldsymbol \Lambda \mathbf{U}^H\mathbf{e}_{ {k}, {l}} =- \mathbf{w}\mathbf{w}^H\hat{\mathbf{g}}_{ {k}, {l}} ,
\end{align}
which means that the optimal solution is given by
\begin{align}
\mathbf{e}_{ {k}, {l}}^*(\lambda^*) =-  \mathbf{U}\boldsymbol \Lambda^{-1}\mathbf{U}^H\mathbf{w}\mathbf{w}^H\hat{\mathbf{g}}_{ {k}, {l}} ,
\end{align}
where $\lambda^*$ is obtained by ensuring that $\mathbf{e}_{ {k}, {l}}^*(\lambda^*) ^H\mathbf{e}_{ {k}, {l}}^*(\lambda^*) =\sigma^2$. We note that a closed-form expression of $\lambda^*$ is difficult to find, but a closed-form expression for the minimum of the objective function in \eqref{eq48} can be found, as shown in the following.

We first   simplify the expression of $\mathbf{e}_{ {k}, {l}}^*(\lambda^*) $. Because of the used decomposition in \eqref{decom}, we have the following property: 
\[
\mathbf{U}^H\mathbf{w} = \begin{bmatrix}\frac{\mathbf{w}^H\mathbf{w}}{|\mathbf{w}^H\mathbf{w}|}&0&\cdots &0 \end{bmatrix}^T.
\]
Therefore, the optimal solution can be further simplified as follows:
\begin{align}
\mathbf{e}_{ {k}, {l}}^*(\lambda^*)  =-  \frac{\mathbf{w}\mathbf{w}^H\hat{\mathbf{g}}_{ {k}, {l}} }{\mathbf{w}^H\mathbf{w}+\lambda^*} .
\end{align}

With this optimal solution, the minimum of the objective function in \eqref{eq48} can be written as follows:
\begin{align}\label{psi}
\psi \triangleq& (\mathbf{e}_{ {k}, {l}}^*(\lambda^*) )^H\mathbf{w}\mathbf{w}^H\mathbf{e}_{ {k}, {l}}^*(\lambda^*) + (\mathbf{e}_{ {k}, {l}}^*(\lambda^*) )^H\mathbf{w}\mathbf{w}^H\hat{\mathbf{g}}_{ {k}, {l}}\\\nonumber &+\hat{\mathbf{g}}_{ {k}, {l}}^H\mathbf{w}\mathbf{w}^H\mathbf{e}_{ {k}, {l}}^*(\lambda^*)+\hat{\mathbf{g}}_{ {k}, {l}}^H\mathbf{w}\mathbf{w}^H \hat{\mathbf{g}}_{ {k}, {l}} .
\end{align}
We first note that the first term in the objective can be expressed as follows:
\begin{align}\label{eqx1}
& (\mathbf{e}_{ {k}, {l}}^*(\lambda^*) )^H\mathbf{w}\mathbf{w}^H\mathbf{e}_{ {k}, {l}} ^*(\lambda^*) \\\nonumber
 = &\frac{\hat{\mathbf{g}}_{ {k}, {l}} ^H\mathbf{w}\mathbf{w}^H}{\mathbf{w}^H\mathbf{w}+\lambda^*} \mathbf{w}\mathbf{w}^H \frac{\mathbf{w}\mathbf{w}^H\hat{\mathbf{g}}_{ {k}, {l}} }{\mathbf{w}^H\mathbf{w}+\lambda^*}  
 = \frac{\hat{\mathbf{g}}_{ {k}, {l}} ^H\mathbf{w}\mathbf{w}^H\mathbf{w}\mathbf{w}^H   \mathbf{w}\mathbf{w}^H\hat{\mathbf{g}}_{ {k}, {l}}  }{(\mathbf{w}^H\mathbf{w}+\lambda^*)^2}. 
\end{align}
In order to obtain an expression without $\lambda^*$, we note that   $ (\mathbf{e}_{ {k}, {l}}^*(\lambda^*) )^H \mathbf{e}_{ {k}, {l}}^*(\lambda^*) - \sigma^{2}  =0 $, which means the following:
 \begin{align}
 \frac{\hat{\mathbf{g}}_{ {k}, {l}}^H\mathbf{w}\mathbf{w}^H }{\mathbf{w}^H\mathbf{w}+\lambda^*}   \frac{\mathbf{w}\mathbf{w}^H\hat{\mathbf{g}}_{ {k}, {l}} }{\mathbf{w}^H\mathbf{w}+\lambda^*}  -\sigma^{2}   =0,
\end{align}
which can be further  written as follows:
 \begin{align}\label{eq60}
(\mathbf{w}^H\mathbf{w}+\lambda^*)^2=\sigma^{-2} \mathbf{w}^H \mathbf{w}\hat{\mathbf{g}}_{ {k}, {l}}^H\mathbf{w}\mathbf{w}^H\hat{\mathbf{g}}_{ {k}, {l}}  .
\end{align}
By substituting \eqref{eq60} in \eqref{eqx1}, the first term in the objective function \eqref{psi} can be simplified as follows:
\begin{align}
 &(\mathbf{e}_{ {k}, {l}}^*(\lambda^*))^H\mathbf{w}\mathbf{w}^H\mathbf{e}_{ {k}, {l}} ^*(\lambda^*)\\\nonumber
  = &\frac{(\mathbf{w}^H\mathbf{w})^2\hat{\mathbf{g}}_{ {k}, {l}} ^H\mathbf{w} \mathbf{w}^H\hat{\mathbf{g}}_{ {k}, {l}}  }{\sigma^{-2} \mathbf{w}^H \mathbf{w}\hat{\mathbf{g}}_{ {k}, {l}}^H\mathbf{w}\mathbf{w}^H\hat{\mathbf{g}}_{ {k}, {l}}  }  = \sigma^{2}  \mathbf{w}^H\mathbf{w} ,
\end{align}
where $\lambda^*$ is removed. 
On the other hand, the third term in the objective function \eqref{psi} can be simplified as follows:
\begin{align}
 \hat{\mathbf{g}}_{ {k}, {l}}^H\mathbf{w}\mathbf{w}^H\mathbf{e}_{ {k}, {l}}^*(\lambda^*)=&-
  \hat{\mathbf{g}}_{ {k}, {l}}^H\mathbf{w}\mathbf{w}^H \frac{\mathbf{w}\mathbf{w}^H\hat{\mathbf{g}}_{ {k}, {l}} }{\mathbf{w}^H\mathbf{w}+\lambda^*} \\\nonumber
  =&-\mathbf{w}^H\mathbf{w}
  \frac{ \hat{\mathbf{g}}_{ {k}, {l}}^H\mathbf{w}\mathbf{w}^H\hat{\mathbf{g}}_{ {k}, {l}} }{\sqrt{\sigma^{-2} \mathbf{w}^H \mathbf{w}\hat{\mathbf{g}}_{ {k}, {l}}^H\mathbf{w}\mathbf{w}^H\hat{\mathbf{g}}_{ {k}, {l}}  }} 
  \\\nonumber
  =&- \sqrt{\sigma^{2} \mathbf{w}^H \mathbf{w}\hat{\mathbf{g}}_{ {k}, {l}}^H\mathbf{w}\mathbf{w}^H\hat{\mathbf{g}}_{ {k}, {l}}  }.
\end{align}
Therefore, the minimum of the objective function \eqref{psi}  is given by
\begin{align}\nonumber
\psi=&  \sigma^{2} \mathbf{w}^H\mathbf{w}  -  2 \sqrt{\sigma^{2} \mathbf{w}^H \mathbf{w}\hat{\mathbf{g}}_{ {k}, {l}}^H\mathbf{w}\mathbf{w}^H\hat{\mathbf{g}}_{ {k}, {l}}  }+\hat{\mathbf{g}}_{ {k}, {l}}^H\mathbf{w}\mathbf{w}^H \hat{\mathbf{g}}_{ {k}, {l}}
 \\ \label{eq63}
 =& \left( \left(\sigma^{2}\mathbf{w}^H\mathbf{w}\right)^{\frac{1}{2}}  -  \left(\mathbf{w}^H \hat{\mathbf{g}}_{ {k}, {l}}\hat{\mathbf{g}}_{ {k}, {l}}^H\mathbf{w}\right)^{\frac{1}{2}}\right)^2.
\end{align}
By substituting \eqref{eq63} in \eqref{pb:6}, the lemma is proved.

\vspace{-1em}
 
   \bibliographystyle{IEEEtran}
\bibliography{IEEEfull,trasfer}

 \begin{IEEEbiography}[{\includegraphics[width=1in,
    height=1.25in,clip, keepaspectratio]{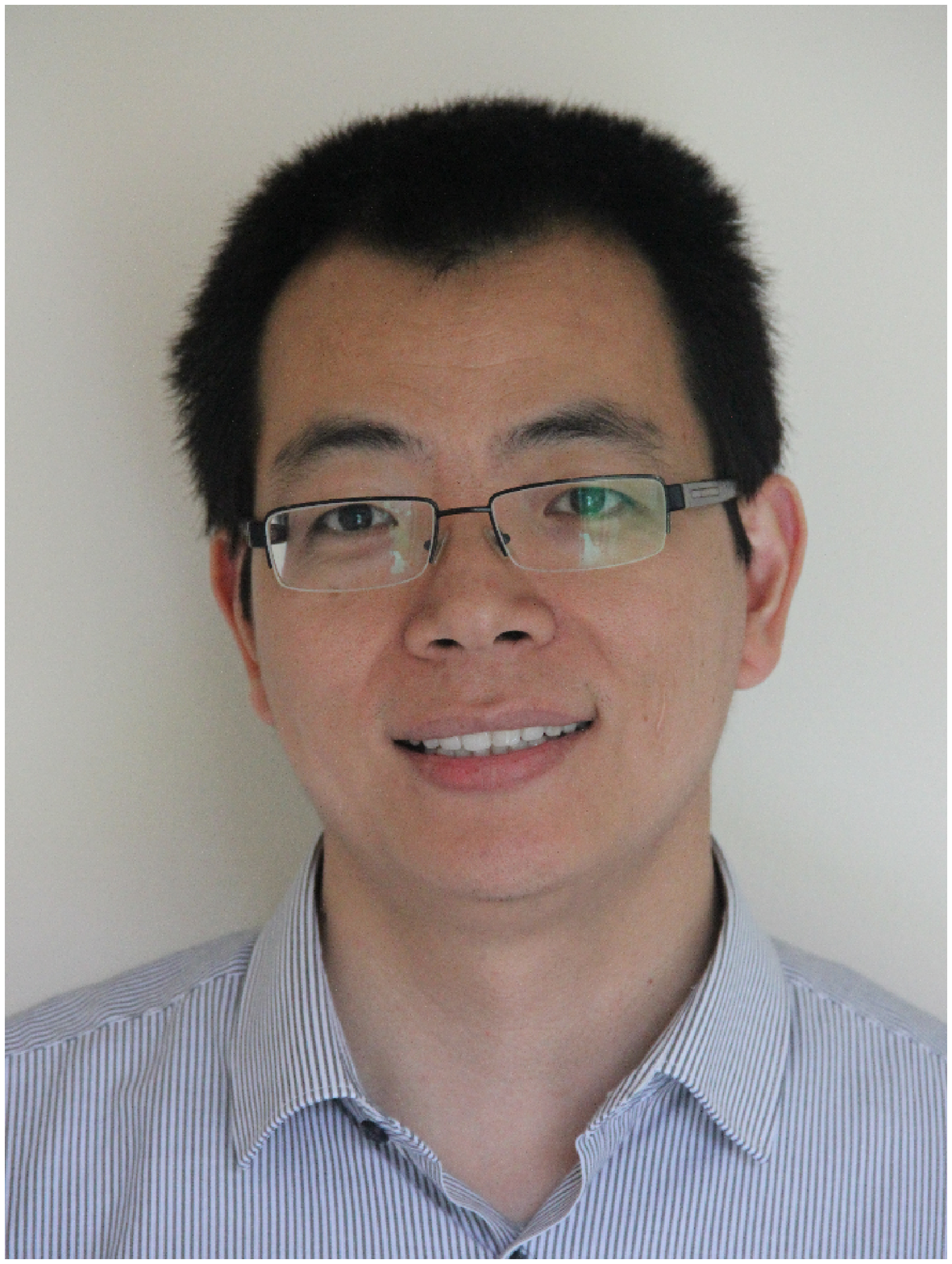}}]{Zhiguo
    Ding} (S'03-M'05-SM'15) received his B.Eng in Electrical Engineering
    from the Beijing University of Posts and Telecommunications in
    2000, and the Ph.D degree in Electrical Engineering from Imperial
    College London in 2005.
      From Jul. 2005 to
    Apr. 2018, he was working in Queen's University Belfast, Imperial
    College, Newcastle University and Lancaster University. Since Apr. 2018, he has been
    with the University of Manchester as a  Professor in Communications. From Oct. 2012 to Sept. 2018, he has also been an academic visitor in Princeton University.
    
    Dr Ding' research interests are 5G networks, game theory, cooperative and energy harvesting networks and statistical signal processing.
    He is serving  as an Area Editor for {\it IEEE Open Journal of the Communications Society}, an Editor
      for {\it IEEE Transactions on Communications} and {\it IEEE Transactions on Vehicular Technology}, and was an Editor for {\it IEEE Wireless Communication Letters}, {\it IEEE Communication Letters} and {\it Journal of Wireless Communications and Mobile Computing}.
      He received the best paper
      award in IET ICWMC-2009 and IEEE WCSP-2014,   the EU Marie Curie Fellowship 2012-2014, the Top IEEE TVT Editor 2017,  2018 IEEE Communication Society Heinrich Hertz Award, 2018 IEEE Vehicular Technology Society Jack Neubauer Memorial Award and 2018 IEEE Signal Processing Society Best Signal Processing Letter Award. \end{IEEEbiography}\vspace{-1em}
   \end{document}